\documentclass[12pt,letterpaper]{article}

\usepackage{amsmath,amsfonts,amsthm,amssymb,stmaryrd,graphicx,relsize,bm}

\theoremstyle{plain}

\numberwithin{equation}{section}
\newtheorem{thm}{Theorem}[section]
\newtheorem{lem}[thm]{Lemma}
\newtheorem{cor}[thm]{Corollary}

\newenvironment{exam}[1]%
{\begin{flushleft}\textbf{Example #1}.\enspace}%
{\end{flushleft}}

\newcommand{\complex}{{\mathbb C}}
\newcommand{\real}{{\mathbb R}}
\newcommand{\pscript}{{\mathcal P}}
\newcommand{\rscript}{{\mathcal R}}
\newcommand{\sscript}{{\mathcal S}}
\newcommand{\itre}{\mathop{\mathit{Re}}}
\newcommand{\rmtr}{\mathop{\mathrm{tr}}}
\newcommand{\atilde}{\widetilde{a}}
\newcommand{\cbar}{\overline{c}}
\newcommand{\abar}{\overline{a}}
\newcommand{\ahat}{\widehat{a}}
\newcommand{\xhat}{\widehat{x}}
\newcommand{\omegahat}{\widehat{\omega}}
\newcommand{\gammahat}{\widehat{\gamma}}

\newcommand{\ab}[1]{\left|#1\right|}
\newcommand{\brac}[1]{\left\{#1\right\}}
\newcommand{\paren}[1]{\left(#1\right)}
\newcommand{\sqbrac}[1]{\left[#1\right]}
\newcommand{\sqparen}[1]{{\left[#1\right)}}
\newcommand{\elbows}[1]{{\left\langle#1\right\rangle}}
\newcommand{\ket}[1]{{\left|#1\right>}}
\newcommand{\bra}[1]{{\left<#1\right|}}

\begin{document}

\title{AN ISOMETRIC DYNAMICS\\FOR A CAUSAL SET APPROACH\\TO DISCRETE QUANTUM GRAVITY
}
\author{S. Gudder\\ Department of Mathematics\\
University of Denver\\ Denver, Colorado 80208, U.S.A.\\
sgudder@du.edu
}
\date{}
\maketitle

\begin{abstract}
We consider a covariant causal set approach to discrete quantum gravity. We first review the microscopic picture of this approach. In this picture a universe grows one element at a time and its geometry is determined by a sequence of integers called the shell sequence. We next present the macroscopic picture which is described by a sequential growth process. We introduce a model in which the dynamics is governed by a quantum transition amplitude. The amplitude satisfies a stochastic and unitary condition and the resulting dynamics becomes isometric. We show that the dynamics preserves stochastic states. By ``doubling down'' on the dynamics we obtain a unitary group representation and a natural energy operator. These unitary operators are employed to define canonical position and momentum operators.
\end{abstract}

\section{Microscopic Picture}  
We call a finite poset $(x,<)$ a \textit{causet} and interpret $a<b$ in $x$ to mean that $b$ is in the causal future of $a$. If $x$ and $y$ are causets with cardinality $\ab{y}=\ab{x}+1$, then $x$ \textit{produces} $y$ (denoted $x\mapsto y$) if $y$ is obtained from $x$ by adjoining a single maximal element to $x$. If $x\to y$ we call $y$ an \textit{offspring} of $x$. A \textit{labeling} for a causet $x$ is a bijection
\begin{equation*}
\ell\colon x\to\brac{1,2,\ldots ,\ab{x}}
\end{equation*}
such that $a,b\in x$ with $a<b$ implies $\ell (a)<\ell (b)$. We then call $x=(x,\ell )$ a \textit{labeled} causet. A labeling of $x$ corresponds to a ``birth order'' for the elements of $x$. Two labeled causets $x,y$ are \textit{isomorphic} if there is a bijection $\phi\colon x\to y$ such that $a<b$ in $x$ if and only if $\phi (a)<\phi (b)$ in $y$ and $\ell\sqbrac{\phi (a)}=\ell (a)$ for all $a\in x$. A causet is \textit{covariant} if it has a unique labeling (up to isomorphisms) and we call a covariant causet a $c$-\textit{causet}. Covariance corresponds to the properties of a manifold being independent of the coordinate system used to describe it. Denote the set of $c$-causets with cardinality $n$ by $\pscript _n$ and the set of all $c$-causets by $\pscript$. It is shown in \cite{gud142} that any $x\in\pscript$ with $x\ne\emptyset$ has a unique producer in $\pscript$ and precisely two offspring in $\pscript$. It follows that $\ab{\pscript _n}=2^{n-1}$, $n=1,2,\ldots\,$. For more background concerning the causet approach to discrete quantum gravity we refer the reader to \cite{hen09,sor94,sur11}. For more information about $c$-causets the reader can refer to \cite{gud141,gud13,gud142}.

Two elements $a,b\in x$ are \textit{comparable} if $a<b$ or $b<a$. We say that $a$ is a \textit{parent} of $b$ and $b$ is a \textit{child} of $a$ if $a<b$ and there is no $c\in x$ with $a<c<b$. A \textit{path} from $a$ to $b$ in $x$ is a sequence $a_1=a,a_2,\ldots ,a_{n-1},a_n=b$ where $a_i$ is a parent of $a_{i+1}$, $i=1,\ldots ,n-1$. The \textit{height} $h(a)$ of $a\in x$ is the cardinality minus one of a longest path in $x$ that ends with $a$. If there is no such path, we set $h(a)=0$. It is shown in \cite{gud142} that a causet $x$ is covariant if and only if $a,b\in x$ are comparable whenever $h(a)\ne h(b)$.

If $x\in\pscript$, we call the sets
\begin{equation*}
S_j(x)=\brac{a\in x\colon h(a)=j}, j=0,1,2,\ldots
\end{equation*}
\textit{shells} and the sequence of integers $s_j(x)=\ab{S_j(x)}$, $j=0,1,2,\ldots$ is the \textit{shell sequence} for $x$ \cite{gud141}. A $c$-causet is uniquely determined by its shell sequence and we think of $\brac{s_j(x)}$ as describing the ``shape'' or geometry of $x$. The tree $(\pscript ,\to )$ can be thought of as a growth model and an $x\in\pscript _n$ is a possible universe at step (time) $n$. An instantaneous universe $x\in\pscript _n$ grows one element at a time in one of two ways. If $x\in\pscript _n$ has shell sequence $(s_0(x),s_1(x),\ldots ,s_m(x))$, then $x\to x_0$ or $x\to x_1$ where $x_0,x_1$ have shell sequence $(s_0(x),s_1(x),\ldots ,s_m(x)+1)$ and $(s_0(x),s_1(x),\ldots ,s_m(x),1)$, respectively. In this way, we recursively order the $c$-causets in $\pscript$ using the notation $x_{n,j}$, $n=1,2,\ldots$, $j=0,1,2,\ldots ,2^{n-1}-1$, where $n=\ab{x_{n,j}}$. For example, in terms of their shell sequences we have:
\begin{align*}  
x_{1,0}&=(1),x_{2,0}=(2),x_{2,1}=(1,1), x_{3,0}=(3), x_{3,1}=(2,1),x_{3,2}=(1,2),x_{3,3}=(1,1,1)\\
x_{4,0}&=(4),x_{4,1}=(3,1),x_{4,2}=(2,2),x_{4,3}=(2,1,1),x_{4,4}=(1,3), x_{4,5}=(1,2,1)\\
x_{4,6}&=(1,1,2), x_{4,7}=(1,1,1,1)
\end{align*}

In the microscopic picture, we view a $c$-causet as a framework or scaffolding for a possible universe. The vertices of $x$ represent small cells that can be empty or occupied by a particle. The shell sequence for $x$ gives the geometry of the framework. In \cite{gud141} we have shown how to construct a metric or distance function on $x$. This metric has simple and useful properties. However, the present paper is mainly devoted to the macroscopic picture and the quantum dynamics that can be developed in that picture. Figure~1 illustrates the first four steps of the sequential growth process $(\pscript ,\to )$. Notice that this is a multiverse model in which infinite paths represent the histories of ``completed'' universes \cite{hen09}.

\begin{figure}
\includegraphics*{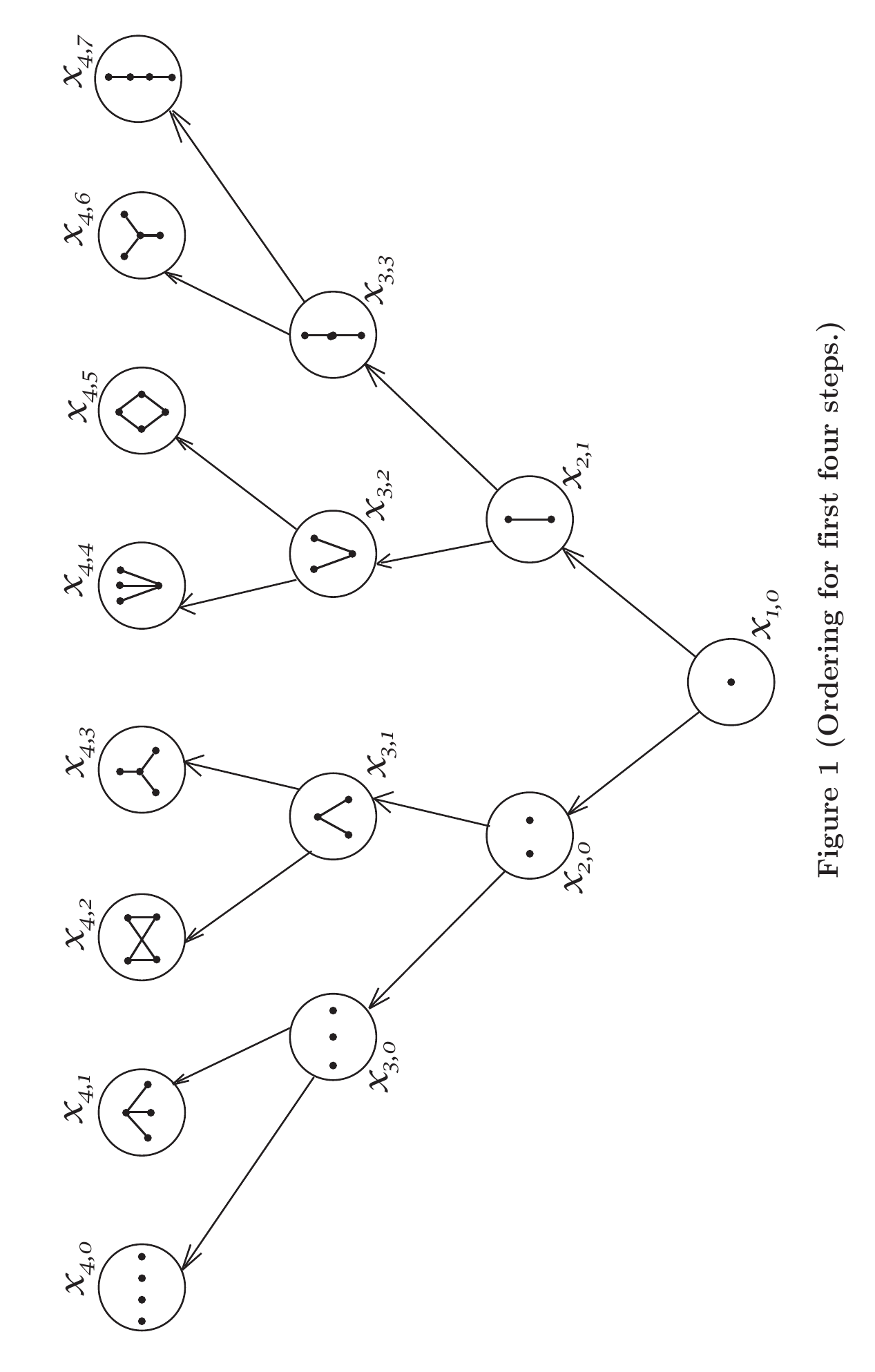}
\end{figure}

\section{Macroscopic Picture} 
We now study the macroscopic picture which describes the evolution of a universe as a quantum sequential growth process. In such a process, the probabilities and propensities of competing geometries are determined by quantum amplitudes. These amplitudes provide interferences that are characteristic of quantum systems. A \textit{transition amplitude} is a map $\atilde\colon\pscript\times\pscript\to\complex$ satisfying
$\atilde (x,y)=0$ if $x\not\to y$ and $\sum _{y\in\pscript}\atilde (x,y)=1$ for every $x\in\pscript$. Since $x_{n,j}$ only has the offspring $x_{n+1,2j}$ and $x_{n+1,2j+1}$ we have that
\begin{equation}         
\label{eq21}
\sum _{k=0}^1\atilde (x_{n,j},x_{n+1,2j+k})=1
\end{equation}
for all $n=1,2,\ldots$, $j=0,1,2,\ldots ,2^{n-1}-1$. We call $\atilde$ a \textit{unitary transition amplitude} (uta) if $\atilde$ also satisfies
$\sum _{y\in\pscript}\ab{\atilde (x,y)}^2=1$ or as in \eqref{eq21} we have
\begin{equation}         
\label{eq22}
\sum _{k=0}^1\ab{\atilde (x_{n,j},x_{n+1,2j+k})}^2=1
\end{equation}
One might suspect that these restrictions on a uta are so strong that the possibilities are very limited. This would be true if $\atilde$ were real valued. In this case, $\atilde (x,y)=1$ for one $y$ with $x\to y$ and $\atilde (x,y)=0$, otherwise. However, in the complex case, the next result shows that there are a continuum of possibilities.
\begin{thm}       
\label{thm21}
Two complex numbers $a,b$ satisfy $a+b=\ab{a}^2+\ab{b}^2=1$ if and only if there exists a $\theta\in\sqparen{0,\pi}$ such that
$a=\cos\theta e^{i\theta}$ and $b=-i\sin\theta e^{i\theta}$. Moreover, $\theta$ is unique.
\end{thm}
\begin{proof}
Necessity is clear. For sufficiency, suppose the conditions $a+b=\ab{a}^2+\ab{b}^2=1$ hold. Then
\begin{equation*}
1=\ab{a}^2+\ab{b}^2=\ab{a}^2+\ab{1-a}^2=\ab{a}^2+(1-a)(1-\abar )=1-2\itre a+2\ab{a}^2
\end{equation*}
Hence, $\ab{a}^2=\itre a$. Letting $a=\ab{a}e^{i\theta}$ we have that $\ab{a}^2=\ab{a}\cos\theta$. If $a=0$, the result holds with $\theta =\pi /2$. If
$a\neq 0$, we have that $\ab{a}=\cos\theta$ and $\itre a=\ab{a}\cos\theta$. Hence, $a=\cos\theta e^{i\theta}$ and
\begin{align*}
b&=1-\cos\theta e^{i\theta}=1-\cos ^2\theta -i\cos\theta\sin\theta =\sin\theta (\sin\theta -i\cos\theta )\\
&=-i\sin\theta e^{i\theta}
\end{align*}
Uniqueness follows from the fact that $\cos\theta$ is injective on $\sqparen{0,\pi}$.
\end{proof}

If $\atilde\colon\pscript\times\pscript\to\complex$ is a uta, we call
\begin{equation*}
c_{n,j}^k=\atilde (x_{n,j},x_{n+1,2j+k}),\quad k=0,1
\end{equation*}
the coupling constants for $\atilde$. It follows from Theorem~\ref{thm21} that there exist $\theta _{n,j}\in\sqparen{0,\pi}$ such that
\begin{equation*}
c_{n,j}^0=\cos\theta _{n,j}e^{i\theta _{n,j}},\quad c_{n,j}^1=-i\sin\theta _{n,j}e^{i\theta _{n,j}}
\end{equation*}
It follows that $c_{n,j}^0+c_{n,j}^1=\ab{c_{n,j}^0}^2+\ab{c_{n,j}^1}^2=1$ for all $n=1,2,\ldots$, $j=0,1,2\ldots ,2^{n-1}-1$. Let $H_n$ be the Hilbert space
\begin{equation*}
H_n=L_2(\pscript _n)=\brac{f\colon\pscript _n\to\complex}
\end{equation*}
with the standard inner product
\begin{equation*}
\elbows{f,g}=\sum _{x\in\pscript _n}\overline{f(x)}g(x)
\end{equation*}
A \textit{path} in $\pscript$ is a sequence $\omega =\omega _1\omega _2\cdots$ where $\omega _i\in\pscript _i$ and $\omega _i\to\omega _{i+1}$
Similarly, an $n$-\textit{path} has the form $\omega =\omega _1\omega _2\cdots\omega _n$ where again $\omega _i\in\pscript _i$ and
$\omega _i\to\omega _{i+1}$. We denote the set of paths by $\Omega$ and the set of $n$-paths by $\Omega _n$. Since every $x\in\pscript _n$ has a unique $n$-path terminating at $x$, we can identify $\pscript _n$ with $\Omega _n$ and we write $\pscript\approx\Omega _n$. Similarly, we identify $H_n$ with $L_2 (\Omega _n)$. If $\atilde$ is a uta and $\omega =\omega _1\omega _2\cdots\omega _n\in\Omega _n$, we define the \textit{amplitude} of $\omega$ to be
\begin{equation*}
a(\omega )=\atilde (\omega _1,\omega _2)\atilde (\omega _2,\omega _3)\cdots\atilde (\omega _{n-1},\omega _n)
\end{equation*}
Moreover, we define the \textit{amplitude} of $x\in\pscript _n$ to be $a(\omega )$ where $\omega\in\Omega _n$ terminates at $x$.

Let $\xhat _{n,}$ be the unit vector in $H_n$ given by the characteristic function $\chi _{x_{n,j}}$. Then clearly,
$\brac{\xhat _{n,j}\colon j=0,1,\ldots ,2^{n-1}-1}$ forms an orthonormal basis for $H_n$. Define the operators $U_n\colon H_n\to H_{n+1}$ by
\begin{equation*}
U_n\xhat _{n,j}=\sum _{k=0}^1c_{n,j}^k\xhat _{n+1,2j+k}
\end{equation*}
and extend $U_n$ to $H_n$ by linearity.

\begin{thm}       
\label{thm22}
{\rm{(i)}}\enspace The adjoint of $U_n$ is given by $U_n^*\colon H_{n+1}\to H_n$ where
\begin{equation}         
\label{eq23}
U_n^*\xhat _{n+1,2j+k}=\cbar _{n,j}^k\xhat _{n,j},\quad k=0,1
\end{equation}
{\rm{(ii)}}\enspace $U_n$ is a partial isometry with $U_n^*U_n=I_n$ and
\begin{equation}         
\label{eq24}
U_nU_n^*=\sum _{j=0}^{2^{n-1}-1}\ket{\sum _{k=0}^1c_{n,j}^k\xhat _{n+1,2j+k}}\bra{\sum _{k=0}^1c_{n,j}^k\xhat _{n+1,2j+k}}
\end{equation}
\end{thm}
\begin{proof}
(i)\enspace To show that \eqref{eq23} holds, we have
\begin{align*}
\elbows{U_n^*\xhat _{n+1,2j'+k'},\xhat _{n,j}}&=\elbows{\xhat _{n+1,2j'+k'},U_n\xhat _{n,j}}\\
  &=\elbows{\xhat _{n+1,2j'+k'},\sum _{k=0}^1c_{n,k}^k\xhat _{n+1,2j+k}}\\
  &=c_{n,j}^k\delta _{jj'}\delta _{kk'}=\elbows{\cbar _{n,j'}^{k'}\xhat _{n,j'}\xhat _{n,j}}
\end{align*}
(ii)\enspace To show that $U_n^*U_n=I_n$ we have by (i) that
\begin{equation*}
U_n^*U_n\xhat _{n,j}=\sum _{k=0}^1c_{n,j}^kU_n^*\xhat _{n+1,2j+k}=\sum _{k=0}^1\ab{c_{n,j}^k}^2\xhat _{n,j}=\xhat _{n,j}
\end{equation*}
Since $\brac{\xhat _{n,j}\colon j=0,1,\ldots ,2^{n-1}-1}$ forms an orthonormal basis for $H_n$, the result follows. Equation~\eqref{eq24} holds because it is well-known that $U_nU_n^*$ is the projection onto the range of $\rscript (U_n)$. We can also show this directly as follows
\begin{align*}
\sum _{j=0}^{2^{n-1}-1}&\ket{\sum _{k=0}^1c_{n,j}^k\xhat _{n+1,2j+k}}\bra{\sum _{k=0}^1c_{n,j}^k\xhat _{n+1,2j+k}}\xhat _{n+1,2j'+k'}\\
  &=\sum _{j=0}^{2^{n-1}-1}\sum _{k=0}^1c_{n,j}^k\xhat _{n+1,2j+k}\cbar_{n,j'}^{k'}\delta _{jj'}=\cbar_{n,j'}^{k'}\sum _{k=0}^1c_{n,j'}^kx_{n+1,2j'+k}\\
  &=\cbar_{n,j'}^{k'}U_n\xhat _{n,j'}=U_nU_n^*\xhat _{n+1,2j'+k'}\qedhere
\end{align*}
\end{proof}

It follows from Theorem~\ref{thm22} that the dynamics $U_n\colon H_n\to H_{n+1}$ for a uta $\atilde$ is an isometric operator. As usual
a \textit{state} on $H_n$ is a positive operator $\rho$ on $H_n$ with $\rmtr (\rho )=1$. A \textit{stochastic state} on $H_n$ is a state $\rho$ that satisfies
$\elbows{\rho 1_n,1_n}=1$ where $1_n=\chi _{\pscript _n}$; that is, $1_n(x)=1$ for every $x\in\pscript _n$. Notice that $U_n^*1_{n+1}=1_n$.

\begin{lem}       
\label{lem23}
{\rm{(i)}}\enspace If $\rho$ is a state on $H_n$, then $U_n\rho U_n^*$ is a state on $H_{n+1}$.
{\rm{(i)}}\enspace If $\rho$ is a stochastic state on $H_n$, then $U_n\rho U_n^*$ is a stochastic state on $H_{n+1}$.
\end{lem}
\begin{proof}
(i)\enspace To show that $U_n\rho _nU_n^*$ is positive, we have
\begin{equation*}
\elbows{U_n\rho U_n^*\phi ,\phi}=\elbows{\rho U_n^*\phi ,U_n^*\phi}\ge 0
\end{equation*}
for all $\phi\in H_{n+1}$. Moreover, by Theorem~\ref{thm22}(ii) we have
\begin{equation*}
\rmtr (U_n\rho U_n^*)=\rmtr (U_n^*U_n\rho )=\rmtr (\rho )=1
\end{equation*}
(ii)\enspace Since $U_n^*1_{n+1}=1_n$ we have
\begin{equation*}
\elbows{U_n\rho U_n^*1_{n+1},1_{n+1}}=\elbows{\rho U_n^*1_{n+1},U_n^*1_{n+1}}=\elbows{\rho 1_n1_n}=1\qedhere
\end{equation*}
\end{proof}

Denoting the time evolution of states by $\rho _n\to\rho _{n+1}$, Lemma~\ref{lem23} shows that $\rho\to U_n\rho U_n^*$ gives a quantum dynamics for states. We now show this explicitly for the transition amplitude. Since
\begin{equation*}
\elbows{\xhat _{n+1,2j+k},U_n\xhat _{n,j}}=c_{n,j}^k=\atilde (x_{n,j},x_{n+1,2j+k})
\end{equation*}
we have for every $\omega =\omega _1\omega _2\cdots\omega _n\in\Omega _n$ that
\begin{equation*}
a(\omega )=\elbows{\omegahat _2,U_1\omegahat _1}\elbows{\omegahat _3,U_2\omegahat _2}
    \cdots\elbows{\omegahat _n,U_{n-1}\omegahat _{n-1}}
\end{equation*}
Define the operator $\rho _n$ on $H_n$ by $\elbows{\omegahat ,\rho _n\omegahat '}=\overline{a(\omega )}a(\omega ')$ where
$\omegahat =\chi _{\brac{\omega}}\in H_n$ for any $\omega\in\Omega _n$.

\begin{thm}       
\label{thm24}
The operator $\rho _n$ is a stochastic state on $H_n$.
\end{thm}
\begin{proof}
To show that $\rho _n$ is positive we have
\begin{align*}
\elbows{f,\rho _nf}&=\elbows{\sum\elbows{\gammahat _i,f}\gammahat _i,\rho _n\sum\elbows{\gammahat _j,f}\gammahat _j}\\
  &=\sum\overline{\elbows{\gammahat _i,f}}\sum\elbows{\gammahat _j,f}\elbows{\gammahat _i,\rho _n\gammahat _j}\\
  &=\sum\overline{\elbows{\gammahat _i,f}}\sum\elbows{\gammahat _j,f}\overline{a(\gamma _i)}a(\gamma _i)\\
  &=\ab{\sum a(\gamma _i)\elbows{\gammahat _i,f}}^2\ge 0
\end{align*}
To show that $\rho _n$ is a state on $H_n$ we have that
\begin{align*}
\rmtr (\rho _n)&=\sum\elbows{\gammahat _i,\rho _n\gammahat _i}=\sum\overline{a(\gamma _i)}a(\gamma _i)=\sum\ab{a(\gamma _i)}^2\\
  &=\sum _{\omega _2}\sum _{\omega _3}\cdots\sum _{\omega _n}\ab{\elbows{\omegahat _2,U_1\omegahat _1}}^2
  \ab{\elbows{\omegahat _3,U_2\omegahat _2}}^2\cdots\ab{\elbows{\omegahat _n,U_{n-1}\omegahat _{n-1}}}^2\\
  &=\sum _{\omega _2}\sum _{\omega _3}\cdots\sum _{\omega _{n-1}}\ab{\elbows{\omegahat _2,U_1\omegahat _1}}^2\\
  &\hskip 4.5pc\cdots\ab{\elbows{\omegahat _{n-1},U_{n-2}\omegahat _{n-2}}}^2
  \sum _{\omega _n}\ab{\elbows{\omegahat _n,U_{n-1}\omegahat _{n-1}}}^2\\
  &=\sum _{\omega _2}\sum _{\omega _3}\cdots\sum _{\omega _{n-1}}\ab{\elbows{\omegahat _2,U_1\omegahat _1}}^2
  \cdots\ab{\elbows{\omegahat _{n-1},U_{n-2}\omegahat _{n-2}}}^2\\
  &\quad\vdots\\
  &=\sum _{\omega _2}\ab{\elbows{\omegahat _2,U_1\omegahat _1}}^2=1
\end{align*}
Finally, $\rho _n$ is stochastic on $H_n$ because
\begin{align*}
\elbows{1_n,\rho _n1_n}&=\elbows{\sum\gammahat _i,\rho\sum\gammahat _j}=\sum _{i,j}\elbows{\gammahat _i,\rho _n\gammahat _j}\\
  &=\sum _{i,j}\overline{a(\gamma _i)}a(\gamma _j)=\ab{\sum a(\gamma _i)}^2
\end{align*}
As before, we obtain
\begin{align*}
\sum _{\omega\in\Omega _n}a(\omega )&=\sum _{\omega _2}\sum _{\omega _3}\cdots\sum _{\omega _n}
  \elbows{\omegahat _2,U_1\omegahat _1}\elbows{\omegahat _3,U_2\omegahat _2}\cdots\elbows{\omegahat _n,U_{n-1}\omegahat _{n-1}}\\
  &=\sum _{\omega _2}\sum _{\omega _3}\cdots\sum _{\omega _{n-1}}
  \elbows{\omegahat _2,U_1\omegahat _1}\elbows{\omegahat _3,U_2,\omegahat _2}\cdots\elbows{\omegahat _{n-1}U_{n-2}\omegahat _{n-2}}\\
  &\quad\vdots\\
  &=\sum _{\omega _2}\elbows{\omegahat _2,U_1\omegahat _1}=1\qedhere
\end{align*}
\end{proof}

If $\omega =\omega _1\omega _2\cdots\omega _n\in\Omega _n$, we have seen that $\omega _n$ produces two offspring
$\omega _{n,0},\omega _{n,1}\in\pscript _{n+1}$. We call the set
\begin{equation*}
(\omega\to )
  =\brac{\omega _1\omega _2\cdots\omega _n\omega _{n,0},\omega _1\omega _2\cdots\omega _n\omega _{n,1}}\subseteq\Omega _{n+1}
\end{equation*}
the \textit{one-step causal future} of $\omega$. We say that the sequence $\rho _n$ is \textit{consistent} if
\begin{equation*}
\elbows{(\omega\to )^\wedge ,\rho _{n+1}(\omega '\to )^\wedge}=\elbows{\omegahat ,\rho _n\omegahat '}
\end{equation*}
for every $\omega ,\omega '\in\Omega _n$ where $(\omega\to )^\wedge =\chi _{(\omega\to )}$. Consistency is important because it follows that the probabilities and propensities given by the dynamics $\rho _n$ are conserved in time \cite{gud13,gud142}.

\begin{thm}       
\label{thm25}
The sequence $\rho _n$ is consistent.
\end{thm}
\begin{proof}
Let $\omega =\omega _1\omega _2\cdots\omega _n$, $\omega '=\omega '_1,\omega '_2\cdots\omega '_n\in\Omega _n$ and suppose that
$\omega _n\to\omega _{n,0},\omega _{n,1}$ and $\omega '_n\to\omega '_{n,0},\omega '_{n,1}$. We then have
\begin{align*}
&\elbows{(\omega\to )^\wedge,\rho _{n+1}(\omega '\to )^\wedge}\\
  &\quad\quad =\left\langle(\omega\omega _{n,0})^\wedge +(\omega\omega _{n,1})^\wedge ,\rho _{n+1}\right.
  \left.\sqbrac{(\omega '\omega '_{n,0})^\wedge+(\omega '\omega '_{n,1})^\wedge}\right\rangle\\
  &\quad\quad =\elbows{(\omega\omega _{n,0})^\wedge ,\rho _{n+1}(\omega '\omega '_{n,0})^\wedge}
  +\elbows{(\omega\omega _{n,0})^\wedge ,\rho _{n+1}(\omega '\omega '_{n,1})^\wedge}\\
  &\quad\quad\quad +\elbows{(\omega\omega _{n,1})^\wedge ,\rho _{n+1}(\omega '\omega '_{n,0})^\wedge}
  \elbows{(\omega\omega _{n,1})^\wedge ,\rho _{n+1}(\omega '\omega '_{n,1})^\wedge}\\
  &\quad\quad =\overline{a(\omega )}\overline{\atilde (\omega _n\omega _{n,0})}a(\omega ')
  \sqbrac{\atilde (\omega '_n,\omega '_{n,0})+\atilde (\omega '_n,\omega '_{n,1})}\\
  &\quad\quad\quad +\overline{a(\omega )}\overline{\atilde (\omega _n,\omega _{n,1})}a(\omega ')
   \sqbrac{\atilde (\omega '_n,\omega '_{n,0})+\atilde (\omega '_n,\omega '_{n,1})}\\
  &\quad\quad =\overline{a(\omega )}a(\omega ')=\elbows{\omegahat ,\rho _n\omegahat '}\qedhere
\end{align*}
\end{proof}

The $n$-\textit{decoherence functional} is the map $D_n\colon 2^{\Omega _n}\times 2^{\Omega _n}\to\complex$ defined by \cite{hen09,sor94,sur11}
\begin{equation*}
D_n(A,B)=\elbows{\chi _B,\rho _n\chi _A}
\end{equation*}
The functional $D_n(A,B)$ gives a measure of the interference between $A$ and $B$ when the system is in the state $\rho _n$. Clearly 
$D_n(\Omega _n,\Omega _n)=1$, $D_n(A,B)=\overline{D_n(A,B)}$ and $A\mapsto D_n(A,B)$ is a complex measure for every $B\in 2^{\Omega _n}$. It is also well-known that if $A_1,\ldots ,A_n\in 2^{\Omega _n}$, then the matrix with entries $D_n(A_j,A_k)$ is positive semidefinite \cite{sor94}. Notice that
\begin{equation*}
D_n\paren{\brac{\omega},\brac{\omega '}}=\overline{a(\omega )}a(\omega ')
\end{equation*}
for every $\omega ,\omega '\in\Omega _n$ and
\begin{equation*}
D(A,B)=\sum\brac{\overline{a(\omega )}a(\omega ')\colon\omega\in A,\omega '\in B}
\end{equation*}
Since $\rho _n$ is consistent, we have that
\begin{equation*}
D_{n+1}\paren{(A\to ),(B\to )}=D_n(A,B)
\end{equation*}
for every $A,B\in 2^{\Omega _n}$ where $(A\to )=\cup\brac{(\omega\to )\colon\omega\in A}$. The corresponding $q$-\textit{measure}
\cite{gud13,sor94,sor03} is the map $\mu _n\colon 2^\Omega\to\real ^+$ defined by
\begin{equation*}
\mu _n(A)=D_n(A,A)=\elbows{\chi _A,\rho _n\chi _A}
\end{equation*}
It follows that $\mu _n(\Omega _n)=1$ and $\mu _{n+1}\paren{(A\to)}=\mu _n(A)$ for all $A\in 2^{\Omega _n}$. Although $\mu _n$ is not additive, it satisfies the \textit{grade} 2-\textit{additive condition}: if $A,B,C\in 2^{\Omega _n}$ are mutually disjoint then \cite{hen09, sor94,sor03,sur11}
\begin{equation*}
\mu _n(A\cup B\cup C)=\mu _n(A\cup B)+\mu _n(A\cup C)+\mu _n(B\cup C)-\mu _n(A)-\mu _n(B)-\mu _n(C)
\end{equation*}
Since $\mu _n$ is not a measure we do call it a probability but we interpret $\mu _n(A)$ as the quantum propensity for the occurrence of $A$. We have discussed in \cite{gud13,gud142} ways of extending the $\mu _n$s to a $q$-measure $\mu$ on suitable subsets of $\Omega$.

A uta is \textit{completely stationary} (cs) \textit{with parameter} $\theta\in\sqparen{0,\pi}$ if $\theta _{n,j}=\theta$ for all $n,j$. For example, let $\atilde$ be cs with parameter $0$. Then the path $x_{1,0}x_{2,0}x_{3,0}\cdots$ has $q$-measure $1$ and all other paths have $q$-measure $0$. Now consider a general cs uta $\atilde$ with parameter $\theta\in (0,\pi )$. When a path ``turns left'' $\atilde$ has the value $\cos\theta e^{i\theta}$ and when it ``turns right'' $\atilde$ has the value $-i\sin\theta e^{i\theta}$. Hence if $\omega\in\Omega _n$ turns left $\ell$ times and right $r$ times we have
\begin{equation*}
a(\omega )=(\cos\theta )^\ell(-i)^r(\sin\theta )^re^{in\theta}
\end{equation*}
We then have
\begin{equation*}
\mu _n\paren{\brac{\omega}}=\ab{a(\omega )}^2=\ab{\cos\theta}^{2\ell}\ab{\sin\theta}^{2r}
\end{equation*}
Hence, $\lim _{n\to\infty}\mu _n\paren{\brac{\omega}}=0$ and is is natural to define $\mu\paren{\brac{\omega}}=0$.

A vector
\begin{equation*}
v=\sum _{j=0}^{2^{n-1}-1}v_j\xhat _{n,j}=\paren{v_o,v_1,\ldots ,v_{2^{n-1}-1}}\in H_n
\end{equation*}
is called a \textit{stochastic state vector} if $\|v\|=1$ and $\elbows{v,1_n}=1$. we call the vector
\begin{equation*}
\ahat_n=\paren{a(x_{n,o}),a(x_{n,1}),\ldots ,a(x_{n,2^{n-1}-1})}\in H_n
\end{equation*}
an \textit{amplitude vector}. Of course, $\ahat _n$ is a stochastic state vector.

\begin{thm}       
\label{thm26}
{\rm{(i)}}\enspace If $v\in H_n$ is a stochastic state vector, then $U_nv\in H_{n+1}$ is also.
{\rm{(ii)}}\enspace $U_n\ahat _n=\ahat _{n+1}$.
{\rm{(iii)}}\enspace $U_n^*\ahat _{n+1}=\ahat _n$.
\end{thm}
\begin{proof}
(i)\enspace This follows from the fact that $U_n$ is isometric and $U_n^*1_{n+1}=1_n$.\newline
(ii)\enspace This holds because
\begin{align*}
U_n\ahat _n&=U_n\sum _{j=0}^{2^{n-1}-1}\!\!a(x_{n,j})\xhat _{n,j}=\!\!\sum _{j=0}^{2^{n-1}-1}\!\!
  \sqbrac{a(x_{n,j})c_{n,j}^0\xhat _{n+1,2j}+a(x_{n,j})c_{n,j}^1\xhat _{n+1,2j+1}}\\
  &=\sum _{j+0}^{2^n-1}a(x_{n+1,j})\xhat _{n+1,j}=\ahat _{n+1}
\end{align*}
(iii)\enspace This is obtained from
\begin{align*}
U_n^*\ahat _{n+1}&=\sum\sqbrac{a(x_{n+1,2j})\xhat _{n+1,2j}+a(x_{n+1,2j+1})\xhat _{n+1,2j+1}}\\
  &=\sum\sqbrac{a(x_{n+1,2j})\cbar _{n,j}^0\xhat _{n,j}+a(x_{n+1,2j+1})\cbar _{n,j}^1\xhat _{n,j}}\\
  &=\sum\sqbrac{c_{n,j}^0a(x_{n,j})\cbar _{n,j}^0+c_{n,j}^1a(x_{n,j})\cbar _{n,j}^1}\xhat _{n,j}\\
  &=\sum a(x_{n,j})\xhat _{n,j}=\ahat _n\qedhere
\end{align*}
\end{proof}

Actually, (iii) follows from (ii) in Theorem~\ref{thm26} because $\ahat _{n+1}=U_n\ahat _n\in\rscript (U_n)$ so
$U_n^*\ahat _{n+1}=U_n^*U_n\ahat _n=\ahat _n$. Interference in $\pscript _n$ or $\Omega _n$ can be described by the nonadditivity of the
$q$-measure $\mu _n$. We say that $x,y\in\pscript _n$ \textit{do not interfere} if
\begin{equation*}
\mu _n\paren{\brac{x,y}}=\mu _n\paren{\brac{x}}+\mu _n\paren{\brac{y}}
\end{equation*}

The next result gives an application of this concept.

\begin{thm}       
\label{thm27}
If $x,y\in\pscript$ have the same producer, then $x$ and $y$ do not interfere.
\end{thm}
\begin{proof}
Suppose $x=x_{n+1,2y},y=x_{n+1,2j+1}$ so $x,y$ have the same producer $x_{n,j}$. Then
\begin{align*}
\mu _{n+1}\paren{\brac{x,y}}&=\ab{a(x)+a(y)}^2=\ab{a(x_{n,j})c_{n,j}^0+a(x_{n,j})c_{n,j}^1}^2\\
&=\ab{a(x_{n,j})}^2\ab{c_{n,j}^0+c_{n,j}^1}^2=\ab{a(x_{n,j})}^2\\
&=\ab{a(x_{n,j})}^2\sqbrac{\ab{c_{n,j}^0}^2+\ab{c_{n,j}^1}^2}\\
&=\ab{a(x_{n,j})c_{n,j}^0}^2+\ab{a(x_{n,j})c_{n,j}^1}^2\\
&=\ab{a(x_{n+1,2j})}^2+\ab{a(x_{n+1,2j+1})}^2=\mu _{n+1}\paren{\brac{x}}+\mu _{n+1}\paren{\brac{y}}
\end{align*}
Hence, $x$ and $y$ do not interfere.
\end{proof}

In general, the noninterference result in Theorem~\ref{thm27} does not hold if $x$ and $y$ have different producers. This is shown in the next two examples.

\begin{exam}{1} 
For simplicity, suppose the uta is cs so we have just two coupling constants $c^0,c^1$. We have seen in Theorem~\ref{thm27} that $x_{3,0}$ and $x_{3,1}$ do not interfere. In a similar way, we see that $x_{3,0}$ and $x_{3,2}$ do not interfere. We also have that $x_{3,1}$ does not interfere with $x_{3,j}$, $j=0,1,2$ and $x_{3,2}$ does not interfere with $x_{3,j}$, $j=0,1,3$. Let us now consider $x_{3,0}$ and $x_{3,3}$. We have that
\begin{align*}
\mu _3\paren{\brac{x_{3,0},x_{3,3}}}&=\ab{a(x_{3,0})+a(x_{3,3})}^2=\ab{(c^0)^2+(c^1)^2}^2\\
  &=\ab{\cos ^2\theta -\sin ^2\theta}=\cos ^22\theta
\end{align*}
On the other hand
\begin{align*}
\mu _3\paren{\brac{x_{3,0}}}+\mu _2\paren{\brac{x_{3,3}}}&=\ab{a(x_{3,0})}^2+\ab{a(x_{3,3}}^2\\
  &=\cos ^4\theta +\sin ^4\theta =\tfrac{1}{2}(1+\cos ^22\theta )
\end{align*}
so $x_{3,0}$ and $x_{3,3}$ interfere, in general.
\end{exam} 

\begin{exam}{2} 
If the uta is not cs, the situation is more complicated and we incur more interference. In the cs case, we saw in Example~1 that $x_{3,0}$ and $x_{3,2}$ do not interfere. However, in this more general case we have
\begin{equation*}
\mu _3\paren{\brac{x_{3,0},x_{3,2}}}=\ab{a(x_{3,0})+a(x_{3,2})}=\ab{c_{1,0}^0c_{2,0}^0+c_{1,0}^1c_{2,1}^0}^2
\end{equation*}
On the other hand
\begin{equation*}
\mu _3\paren{\brac{x_{3,0}}}+\mu _3\paren{\brac{x_{3,2}}}\!=\!\ab{a(x_{3,0})}^2+\ab{a(x_{3,2})}^2
  \!=\!\ab{c_{1,0}^0}^2\ab{c_{2,0}^0}^2+\ab{c_{1,0}^1}^2\ab{c_{2,1}^0}^2
\end{equation*}
But these two quantities do not agree unless
\begin{equation*}
\itre (c_{1,0}^0c_{2,0}^0\cbar _{1,0}^1\cbar _{2,1}^0)=0
\end{equation*}
so $x_{3,0}$ and $x_{3,3}$ interfere, in general.
\end{exam} 

\section{Double-Down To Unitary} 
We have seen that corresponding to a uta with coupling constants $c_{n,j}^k$ there are isometries $U_n\colon H_n\to H_{n+1}$ that describe the dynamics for a quantum sequential growth process on $(\pscript ,\to )$. The operators $U_n$ cannot be unitary because $H_n$ and $H_{n+1}$ are different dimensional Hilbert spaces. However, we can ``double-down'' the $U_n$ to form operators $V_{n+1}\colon H_{n+1}\to H_{n+1}$ by
\begin{align*}
V_{n+1}\xhat _{n+1,2j}&=c_{n,j}^0\xhat _{n+1,2j}+c_{n,j}^1\xhat _{n+1,2j+1}\\
V_{n+1}\xhat _{n+1,2j+1}&=c_{n,j}^1\xhat _{n+1,2j}+c_{n,j}^0\xhat _{n+1,2j+1}\\
\end{align*}

\begin{thm}       
\label{thm31}
The operators $V_{n+1}$ are unitary and $V_{n+1}1_{n+1}=1_{n+1}$, $n=1,2,\ldots\,$.
\end{thm}
\begin{proof}
Since $\|V_{n+1}\xhat _{n+1,2j}\|=\|V_{n+1}\xhat _{n+1,2j+1}\|=1$ and
\begin{equation*}
\elbows{V_{n+1}\xhat _{n+1,2j},V_{n+1}\xhat _{n+1,2j+1}}=\cbar _{n,j}^0c_{n,j}^1+\cbar _{n,j}^1c_{n,j}^0=0
\end{equation*}
we conclude that $V_{n+1}$ sends an orthonormal basis to an  orthonormal basis. Hence, $V_{n+1}$ is unitary. To show that $V_{n+1}1_{n+1}=1_{n+1}$ we have
\begin{align*}
V_{n+1}1_{n+1}&=\sum _j(V_{n+1}\xhat _{n+1,2j}+V_{n+1}\xhat _{n+1,2j+1})\\
&=\sum _j\sqbrac{(c_{n,j}^0+c_{n,j}^1)\xhat _{n+1,2j}+(c_{n,j}^1+c_{n,j}^0)\xhat _{n+1,2j+1}}\\
&=\sum _j(\xhat _{n+1,2j}+\xhat _{n+1,2j+1})=1_{n+1}\qedhere
\end{align*}
\end{proof}

The unitary operator $V_2$ corresponds to the coupling constants $c_{1,0}^0$, $c_{1,0}^1$ and relative to the basis $\brac{\xhat _{2,0},\xhat _{2,1}}$ has the form
\begin{equation*}
V_2=
\begin{bmatrix}
  \noalign{\smallskip}c_{1,0}^0&c_{1,0}^1\\\noalign{\smallskip}
  c_{1,0}^1&c_{1,0}^0\\\noalign{\smallskip}
\end{bmatrix}
\end{equation*}
Besides being unitary, $V_2$ is doubly stochastic (row and column sums are one). Of course, this is also true of $V_n$. By Theorem~\ref{thm21}, there exists a unique $\theta\in\sqparen{0,\pi}$ such that $c_{1,0}^0=\cos\theta e^{i\theta}$, $c_{1,0}^1=-i\sin\theta e^{i\theta}$. To make $\theta$ explicit, we write $V_2=V_2(\theta )$.

\begin{lem}       
\label{lem32}
The operator $V_2(\theta )$ has eigenvalues $1,e^{2i\theta}$ with corresponding unit eigenvectors $2^{-1/2}(1,1)$, $2^{-1/2}(1,-1)$.
\end{lem}
\begin{proof}
By direct verification we have
\begin{align*}
V_2\begin{bmatrix}1\\ 1\end{bmatrix}
&=\begin{bmatrix}\noalign{\smallskip}c_{1,0}^0+c_{1,0}^1\\\noalign{\smallskip} c_{1,0}^1+c_{1,0}^0\\\noalign{\smallskip}\end{bmatrix}
=\begin{bmatrix}1\\ 1\end{bmatrix}\\\noalign{\medskip} 
V_2\begin{bmatrix}1\\ -1\end{bmatrix}
&=\begin{bmatrix}\noalign{\smallskip}c_{1,0}^0-c_{1,0}^1\\\noalign{\smallskip} c_{1,0}^1-c_{1,0}^0\\\noalign{\smallskip}\end{bmatrix}
=(c_{1,0}^0-c_{1,0}^1)\begin{bmatrix}1\\ -1\end{bmatrix}
\end{align*}
But
\begin{align*}
c_{1,0}^0-c_{1,0}^1&=c_{1,0}^0-(1-c_{1,0}^0)=2c_{1,0}^0-1=2\cos\theta e^{i\theta}-1\\
&=2\cos ^2\theta +2i\cos\theta\sin\theta -1=\cos ^2\theta -\sin ^2\theta +i\sin 2\theta\\
&=\cos 2\theta +i\sin 2\theta =e^{2i\theta}\qedhere
\end{align*}
\end{proof}

We can write the $2^n$-dimensional Hilbert space $H_{n+1}$ as
\begin{equation*}
H_{n+1}=H_2\oplus H_2\oplus\cdots\oplus H_2
\end{equation*}
where there are $2^{n-1}$ summands and the $j$th summand has the basis$\brac{\xhat _{n+1,2j},\xhat _{n+1,2j+1}}$. In general, $V_{n+1}$ has the form
\begin{equation*}
V_{n+1}(\theta _1,\theta _2,\ldots ,\theta _{2^{n-1}})=V_2(\theta _1)\oplus V_2(\theta _2)\oplus\cdots\oplus V(\theta _{2^{n-1}})
\end{equation*}
It follows from Lemma~\ref{lem32} that $V_{n+1}(\theta _1,\theta _2,\ldots ,\theta _{2^{n-1}})$ has eigenvalues $1$ (with multiplicity $2^{n-1}$) and
$e^{2i\theta _1},e^{2i\theta _2},\ldots ,e^{2i\theta _{2^{n-1}}}$. The unit eigenvectors corresponding to $1$ are
\begin{equation*}
2^{-1/2}(\xhat _{n+1,2j}+\xhat _{n+1,2j+1}),\quad j=0,1,\ldots ,2^{n-1}-1
\end{equation*}
and the unit eigenvector corresponding to $e^{2i\theta j}$ is
\begin{equation*}
2^{-1/2}(\xhat _{n+1,2j}-\xhat _{n+1,2j+1})
\end{equation*}
Let $\sscript (H_{n+1})$ be the set of operators on $H_{n+1}$ of the form
\begin{equation*}
\sscript (H_{n+1})=\brac{V_{n+1}(\theta _1,\theta _2,\ldots ,\theta _{2^{n-1}})\colon\theta _n\in\sqparen{0,\pi}}
\end{equation*}
Now $\sqparen{0,\pi}$ forms an abelian group with operations $a\oplus b=a+b\pmod\pi$.

\begin{lem}       
\label{thm33}
For $\theta _1,\theta _2\in\sqparen{0,\pi}$ we have $V_2(\theta _1)V_2(\theta _2)=V_2(\theta _1+\theta _2)$.
\end{lem}
\begin{proof}
Since $V_2(\theta _1)$ and $V_2(\theta _2)$ have the same eigenvectors, they commute and can be simultaneously diagonalized as
\begin{equation*}
V_2(\theta _1)=
\begin{bmatrix}
  \noalign{\smallskip}1&0\\\noalign{\smallskip}
  0&e^{2i\theta _1}\\\noalign{\smallskip}
\end{bmatrix}\quad
V_2(\theta _2)=
\begin{bmatrix}
  \noalign{\smallskip}1&0\\\noalign{\smallskip}
  0&e^{2i\theta _2}\\\noalign{\smallskip}
\end{bmatrix}
\end{equation*}
Hence, if $\theta _1+\theta _2<\pi$ then
\begin{equation*}
V_2(\theta _1)V_2(\theta _2)=
\begin{bmatrix}
  \noalign{\smallskip}1&0\\\noalign{\smallskip}
  0&e^{2i(\theta _1+\theta _2)}\\\noalign{\smallskip}
\end{bmatrix}=V_2(\theta _1\oplus\theta _2)
\end{equation*}
and if $\theta _1+\theta _2\ge\pi$ then $\theta _1\oplus\theta _2=\theta _1+\theta _2-\pi$ and we have
\begin{equation*}
V_2(\theta _1)V_2(\theta _2)=
\begin{bmatrix}
  \noalign{\smallskip}1&0\\\noalign{\smallskip}
  0&e^{2i(\theta _1+\theta _2-\pi)}\\\noalign{\smallskip}
\end{bmatrix}=V_2(\theta _1+ \theta _2-\pi)=V(\theta _1\oplus\theta _2)\qedhere
\end{equation*}
\end{proof}

We now form the product group $\sqparen{0,\pi}^{2^{n-1}}=\sqparen{0,\pi}\times\cdots\times\sqparen{0,\pi}$ to obtain the following result.

\begin{cor}       
\label{thm34}
Under operator multiplication, $\sscript (H_{n+1})$ is an abelian group and
$(\theta _1,\ldots ,\theta _{2^{n-1}})\mapsto V_{n+1}(\theta _1,\ldots ,\theta _{2^{n-1}})$ is a unitary representation of the group
$\sqparen{0,\pi}^{2^{n-1}}$.
\end{cor}

Since $V_{n+1}$ is unitary, there exists a unique self-adjoint operator $K_{n+1}$ on $H_{n+1}$ such that $V_{n+1}=e^{iK_{n+1}}$. We call $K_{n+1}$ a \textit{Hamiltonian operator}. For $V_{n+1}(\theta _1,\ldots ,\theta _{2^{n-1}}$ the eigenvalues of $K_{n+1}$ are $0$ (with multiplicity $2^{n-1}$) and $2\theta _1,\ldots ,2\theta _{2^{n-1}}$. Hence, $\theta _j=2^{-1}\lambda _j$ where $\lambda _j$ is the $j$th energy value, $j=1,\ldots ,2^{n-1}$. This gives a physical significance for the angles $\theta _j$. The corresponding eigenvectors are the same as those given for $V_{n+1}$.

It is natural to define the \textit{position operator} $Q_{n+1}$ on $H_{n+1}$ by $Q_{n+1}f(\xhat _{n+1,k})=k$. Thus, $Q_{n+1}\xhat _{n+1,2j}=2j$ and
$Q_{n+1}\xhat _{n+1,2j+1}=2j+1$. Since $Q_{n+1}$ is diagonal, we immediately see that its eigenvalues are $0,1,\ldots ,2^n-1$ with corresponding eigenvector $\xhat _{n+1,k}$. It also seems natural to define the \textit{canonical momentum operator} $P_{n+1}$ on the subspace generated by
$\brac{\xhat _{n+1,2j},\xhat _{n+1,2j+1}}$ as
\begin{align*}
P_2(\theta _j)&=V_2(\theta _j)^*Q_2(\theta _j)V_2(\theta _j)\\
&=\begin{bmatrix}\noalign{\smallskip}\cbar _{n,j}^0&\cbar _{n,j}^1\\\noalign{\smallskip}\cbar _{n,j}^1&\cbar _{n,j}^0\\\noalign{\smallskip}\end{bmatrix}
\quad\begin{bmatrix}\noalign{\smallskip}2j&0\\\noalign{\smallskip}0&2j+1\\\noalign{\smallskip}\end{bmatrix}
\quad\begin{bmatrix}\noalign{\smallskip}c_{n,j}^0&c_{n,j}^1\\\noalign{\smallskip}c_{n,j}^1&c_{n,j}^0\\\noalign{\smallskip}\end{bmatrix}\\
&=\begin{bmatrix}\noalign{\smallskip}2j+\ab{c_{n,j}^1}^2&c_{n,j}^0\cbar _{n,j}^1\\\noalign{\smallskip}
  \cbar _{n,j}^0c_{n,j}^1&2j+\ab{c_{n,j}^0}^2\\\noalign{\smallskip}\end{bmatrix}
  =\begin{bmatrix}\noalign{\smallskip}2j+\sin ^2\theta _{n,j}&\tfrac{i}{2}\sin 2\theta _{n,j}\\\noalign{\smallskip}
 -\tfrac{i}{2}\sin 2\theta _{n,j}&2j+\cos ^2\theta _{n,j}\\\noalign{\smallskip}\end{bmatrix}\\
\end{align*}
The eigenvalues of $P_2(\theta _j)$ are $2j$ and $2j+1$ with corresponding unit eigenvectors
\begin{align*}
V_2(\theta _j)^*\begin{bmatrix}1\\ 0\end{bmatrix}
&=\begin{bmatrix}\noalign{\smallskip}\cbar _{n,j}^0 \\\cbar _{n,j}^1\\\noalign{\smallskip}\end{bmatrix}\\\noalign{\medskip}
V_2(\theta _j)^*\begin{bmatrix}0\\ 1\end{bmatrix}
&=\begin{bmatrix}\noalign{\smallskip}\cbar _{n,j}^1 \\\cbar _{n,j}^0\\\noalign{\smallskip}\end{bmatrix}
\end{align*}
The complete momentum operator $P_{n+1}$ is given by
\begin{equation*}
P_{n+1}(\theta _1,\ldots ,\theta _{2^{n-1}})=P_2(\theta _1)\oplus P_2(\theta _2)\oplus\cdots\oplus P_2(\theta _{2{n-1}})
\end{equation*}

We now compute the commutator
\begin{align*}
\sqbrac{P_2(\theta _j),Q_2(\theta _j)}&=P_2(\theta _j)Q_2(\theta _j)-Q_2(\theta _j)P_2(\theta _n)=c_{n,j}^0\cbar _{n,j}^1
  \begin{bmatrix}\noalign{\smallskip}0&1\\\noalign{\smallskip}1&0\\\noalign{\smallskip}\end{bmatrix}\\
  &=\tfrac{i}{2}\sin 2\theta _j  \begin{bmatrix}\noalign{\smallskip}0&1\\\noalign{\smallskip}1&0\\\noalign{\smallskip}\end{bmatrix}
\end{align*}
The complete commutation relation is
\begin{align*}
&\sqbrac{P_{n+1}(\theta _1,\ldots ,\theta _{2^{n-1}}),Q_{n+1}(\theta _1,\ldots ,\theta _{2^{n-1}})}\\
  &\hskip 4pc =\sqbrac{P_2(\theta _1),Q_2(\theta _1)}\oplus\cdots\oplus\sqbrac{P_2(\theta _{2^{n-1}})Q_2(\theta _{2^{n-1}})}
\end{align*}
As in the Heisenberg uncertainty relation, the number $\ab{\elbows{\phi ,\sqbrac{P_{n+1},Q_{n+1}}\phi}}$ gives a lower bound for the product of the variances of $P_{n+1}$ and $Q_{n+1}$. We now compute this number for an amplitude state $\ahat _{n+1}$. We have that
\begin{align*}
&\elbows{\ahat _{n+1},\sqbrac{P_{n+1},Q_{n+1}}\ahat _{n+1}}\\
&\hskip 3pc =\sum _j\left\langle\begin{bmatrix}\noalign{\smallskip}a(x_{n+1,2j})\\a(x_{n+1,2j+1})\\\noalign{\smallskip}\end{bmatrix},\right.
\left. c_{n,j}^0c_{n,j}^1\begin{bmatrix}\noalign{\smallskip}0&1\\\noalign{\smallskip}1&0\\\noalign{\smallskip}\end{bmatrix}\ \right.
\left. \begin{bmatrix}\noalign{\smallskip}a(x_{n+1,2j})\\a(x_{n+1,2j+1})\\\noalign{\smallskip}\end{bmatrix}\right\rangle\\\noalign{\medskip}
&\hskip 3pc =\sum _jc_{n,j}^0\cbar _{n,j}^1\sqbrac{\abar (x_{n+1,2j})a(x_{n+1,2j+1})+\abar (x_{n+1,2j+1})a(x_{n+1,2j})}\\
&\hskip 3pc =\sum _jc_{n,j}^0\cbar _{n,j}^1\ab{a(x_{n,j})}^2\sqbrac{\cbar _{n,}^0c_{n,j}^1+\cbar _{n,j}^1c_{n,j}^0}=0
\end{align*}
This shows that even though $P_{n+1}$ and $Q_{n+1}$ do not commute, there is no lower bound for the product of their variances when the system is in an amplitude state.

\end{document}